\documentclass[12pt]{l4dc2021} 

\title[Maximum Likelihood Signal Matrix Model for Data-Driven Predictive Control]{Maximum Likelihood Signal Matrix Model for Data-Driven \\Predictive Control}
\usepackage{times}
\usepackage{empheq}
\newcommand{\norm}[1]{\left\lVert#1\right\rVert}
\newcommand{\pazocal}{\mathcal}

\newtheorem{prop}{Proposition}
\newtheorem{ex}{Example}

\author{\Name{Mingzhou Yin} \Email{myin@control.ee.ethz.ch}\\
  \Name{Andrea Iannelli} \Email{iannelli@control.ee.ethz.ch}\\
  \Name{Roy S. Smith} \Email{rsmith@control.ee.ethz.ch}\\
  \addr Automatic Control Laboratory, Swiss Federal Institute of Technology, 8092 Zurich, Switzerland}

\begin{document}

\maketitle\thispagestyle{empty}

\vspace{-2em}

\fbox{\begin{minipage}{0.9\linewidth}
Mingzhou Yin, Andrea Iannelli, Roy S. Smith. Maximum Likelihood Signal Matrix Model for Data-Driven Predictive Control. \textit{Proceedings of the 3rd Conference on Learning for Dynamics and Control}, PMLR 144:1004-1014, 2021.

\copyright\ 2021 the authors. This work has been presented at the 3rd Conference on Learning for Dynamics and Control and published under a Creative Commons Licence CC-BY-NC-ND.
\end{minipage}}

\vspace{2em}

\begin{abstract}%
The paper presents a data-driven predictive control framework based on an implicit input-output mapping derived directly from the signal matrix of collected data. This signal matrix model is derived by maximum likelihood estimation with noise-corrupted data. By linearizing online, the implicit model can be used as a linear constraint to characterize possible trajectories of the system in receding horizon control. The signal matrix can also be updated online with new measurements. This algorithm can be applied to large datasets and slowly time-varying systems, possibly with high noise levels. An additional regularization term on the prediction error can be introduced to enhance the predictability and thus the control performance. Numerical results demonstrate that the proposed signal matrix model predictive control algorithm is effective in multiple applications and performs better than existing data-driven predictive control algorithm.
\end{abstract}


\section{Introduction}

In control applications, despite the wide use of data to understand the response of a system, conventional methods are focused on identifying a compact model of the system based on first principle model structures or low-order approximation (\cite{LjungBook2}). This enables the design of simple but effective feedback control laws. However, this idea of model-based control is challenged by the increasing complexity of modern systems, where compact models are hard to obtain (\cite{Hjalmarsson_2005}). On the other hand, modern systems are usually associated with a large set of available data. Therefore, model-free approaches have been proposed to learn the control action directly from collected data without an explicit formulation of the system model, e.g., reinforcement learning (\cite{lagoudakis2003least,powell2007approximate}). However, it has been observed in \cite{Recht_2019} that these approaches tend to perform in a much less data-efficient manner in simple tasks compared to model-based methods.

In this paper, we focus on a control design strategy which lies in between the model-based and the model-free approaches. This strategy still aims to obtain a description of the system, but in the form of a non-parametric input-output mapping instead of a compact parametric model. In detail, the input-output mapping is derived based on the Willems’ fundamental lemma proposed in \cite{Willems_2005} that characterizes all possible trajectories of length up to the order of persistency of excitation by constructing Hankel matrices from data. This implicit representation of the system has been successfully applied to simulation and control problems with noise-free data in \cite{Markovsky_2008,Coulson_2019,DePersis_2020}. These results in the noise-free case are extended to noise-corrupted datasets by maximum likelihood estimation in \cite{yin2020maximum}. The derived data-driven input-output mapping, dubbed the signal matrix model (SMM), provides a statistically optimal formulation to predict the response of the system.

This paper focuses on the application of SMM to receding horizon predictive control algorithms. This optimization-based control algorithm is suitable to solve optimal trajectory tracking problems with constraints. When a model of the system is available, it is commonly known as model predictive control (MPC) (\cite{camacho2016model}). In this work, the model-based predictor is replaced with a linearized SMM constraint obtained directly from the data-based signal matrix. This algorithm is efficient in handling large datasets by data preconditioning without increasing online computational complexity. It is demonstrated by numerical examples that this algorithm performs better than existing data-driven predictive control algorithm with noisy data. This paper proposes multiple novel extensions of the algorithm and analyzes their performances with numerical examples. Online data can be incorporated in the signal matrix to make the algorithm adaptive. This improves our knowledge of the system online, and extend the algorithm to slowly time-varying systems. When data with high noise levels are treated, better initial condition estimation can be achieved with longer past trajectory measurements. The predictability along the closed-loop trajectories is enhanced by introducing a prediction error regularization term in the objective function, and thus the control performance can be improved as well.


\section{Data-Driven Predictive Control}

Consider a minimal discrete-time linear time-invariant (LTI) dynamical system with output noise, given by
\begin{equation}
\begin{cases}
x_{t+1}&=\ A x_t+B u_t,\\
\hfil y_t&=\ C x_t + D u_t + w_t,
\label{eq:sys}
\end{cases}
\end{equation}
where $x_t \in \mathbb{R}^{n_x}$, $u_t \in \mathbb{R}^{n_u}$, $y_t \in \mathbb{R}^{n_y}$, $w_t \in \mathbb{R}^{n_y}$ are the states, inputs, outputs, and output noise respectively. The model parameters $(A,B,C,D)$ are unknown, but input-output trajectory data $(u_i^d, y_i^d)_{i=0}^{N-1}$ of the system have been collected.

We are interested in designing a data-driven controller that obtains optimal tracking to a reference trajectory $\mathbf{r}$ in the sense of minimizing the total cost:\begin{equation}
    J_{\text{tot}}=\sum_{t=0}^{N_c-1} J_t=\sum_{t=0}^{N_c-1}\left(\norm{y^0_t-r_t}_Q^2+\norm{u_t}_R^2\right),
\end{equation}
where $J_t$ is the stage cost at time $t$, $y^0_t$ is the noise-free output, $Q$ and $R$ are the output and the input cost matrices respectively, and $N_c$ is the length of the control task. This can be achieved by receding horizon predictive control. In detail, at each time instant $t$, the controller predicts an admissible input-output trajectory of length $L'$ of the system that minimizes the predicted cost of future $L'$ steps online, i.e.,
\begin{subequations}
  \begin{empheq}{align}\underset{\hat{\mathbf{u}},\hat{\mathbf{y}}}{\text{minimize}}\quad & J_{\text{ctr}}:=\sum_{k=0}^{L'-1}\left(\norm{\hat{y}_k-r_{t+k}}_Q^2+\norm{\hat{u}_k}_R^2\right)\\
    \text{subject to}\quad & \left(\begin{bmatrix}\mathbf{u}_p\\\hat{\mathbf{u}}\end{bmatrix},\begin{bmatrix}\mathbf{y}_p\\\hat{\mathbf{y}}\end{bmatrix}\right)\text{ is a possible trajectory of (\ref{eq:sys})},\label{eqn:4b}\\&\hat{\mathbf{u}} \in \mathcal{U}, \hat{\mathbf{y}} \in \mathcal{Y},
    \end{empheq}
    \label{eqn:pc}%
\end{subequations}
where $\mathbf{u}_p=(u_i)_{i=-\infty}^{t-1}$, $\mathbf{y}_p=(y_i)_{i=-\infty}^{t-1}$ are the immediate past input-output trajectory, and $\mathcal{U}	\subseteq\mathbb{R}^{L'n_u}$, $\mathcal{Y}\subseteq\mathbb{R}^{L'n_y}$ are the sets of admissible inputs and outputs respectively. Then, the first entry in the newly optimized input sequence is applied to the system in a receding horizon fashion.

In standard model predictive control, the constraint (\ref{eqn:4b}) is enforced based on a known model of the system, such as (\ref{eq:sys}) with a state estimation scheme. Alternatively, in the data-driven predictive control framework, this constraint is directly formulated from the collected data of the systems, namely $(u_i^d, y_i^d)_{i=0}^{N-1}$, without knowing the model parameters. In this regard, it is proved in \cite{Willems_2005} and, in a state-space context, \cite{DePersis_2020,vanWaarde_2020} that all the possible trajectories of a linear system within a time horizon can be captured by a single persistently exciting trajectory of the system when no noise is present. Based on these results, the following proposition holds.

\begin{prop}
    Assume that 1) the system is noise-free, i.e., $\forall i, w_i=0$; 2) the collected input trajectory is persistently exciting of order $(L+n_x)$, i.e., the block Hankel matrix constructed from $(u_i^d)_{i=0}^{N-1}$ with $(L+n_x)$ block rows
    has full row rank; and 3) $L_0:=L-L'\geq n_x$.
    Then constraint (\ref{eqn:4b}) is equivalent to the following statement: there exists $g$, such that
    \begin{equation}
    \begin{bmatrix}\mathbf{u}_{\textnormal{ini}}(t)\\\hat{\mathbf{u}}(t)\\\hline\mathbf{y}_{\textnormal{ini}}(t)\\\hat{\mathbf{y}}(t)\end{bmatrix}
    =\begin{bmatrix}U\\Y\end{bmatrix}g(t):=\begin{bmatrix}
        u_0^d&u_1^d&\cdots&u_{M-1}^d\\
        \vdots&\vdots&\ddots&\vdots\\
        u_{L-1}^d&u_L^d&\cdots&u_{N-1}^d\\\hline
        y_0^d&y_1^d&\cdots&y_{M-1}^d\\
        \vdots&\vdots&\ddots&\vdots\\
        y_{L-1}^d&y_L^d&\cdots&y_{N-1}^d\\
        \end{bmatrix}g(t),
        \label{eqn:fund}
    \end{equation}
    where $\mathbf{u}_{\textnormal{ini}}(t)=(u_i)_{i=t-L_0}^{t-1}$, $\mathbf{y}_{\textnormal{ini}}(t)=(y_i)_{i=t-L_0}^{t-1}$, and $M=N-L+1$.
    \label{prop:1}
\end{prop}
\begin{proof}
    This is a direct extension of Proposition 1 in \cite{Markovsky_2008}.
\end{proof}
Equation (\ref{eqn:fund}) provides an implicit input-output mapping between $\hat{\mathbf{u}}$ and $\hat{\mathbf{y}}$ with an intermediate parameter $g$, by constructing the matrix $\text{col}\left(U,Y\right)$. This matrix is constructed from offline measurement data and referred to as the signal matrix in the remainder of the paper. By replacing the constraint (\ref{eqn:4b}) with the parametric relation (\ref{eqn:fund}), the predictive controller can be designed without any knowledge of the model parameters in (\ref{eq:sys}). This method is known as the data-enabled predictive control (DeePC) algorithm (\cite{Coulson_2019}).

\section{The Maximum Likelihood Signal Matrix Model}

A strong assumption in Proposition~\ref{prop:1} is that the data are noise-free, under which the signal matrix is rank deficient with $\text{rank}\left(\text{col}\left(U,Y\right)\right)=n_x+n_u L$ as shown in Theorem 2 of \cite{Moonen_1989}. This guarantees that, though (\ref{eqn:fund}) is a highly underdetermined linear system when a large dataset is available, the input-output mapping from $\hat{\mathbf{u}}$ to $\hat{\mathbf{y}}$ is still well-defined. However, when the data are noise-corrupted, the signal matrix has full row rank almost surely, and thus for any $\hat{\mathbf{u}},\hat{\mathbf{y}}$, there exists $g$ that satisfies (\ref{eqn:fund}) almost surely. The input-output mapping becomes ill-conditioned.

In order to extend the DeePC algorithm to noise-corrupted data, empirical regularization terms on $g$ have been introduced in the objective function. The optimization problem becomes
\begin{equation}
  \begin{aligned}\underset{\hat{\mathbf{u}},\hat{\mathbf{y}},g}{\textnormal{minimize}}\quad & J_{\textnormal{reg}}:= J_{\text{ctr}} +\lambda_g\norm{g}_p^p+\lambda_y\norm{Y_p g-\mathbf{y}_{\text{ini}}}_p^p\\
    \textnormal{subject to}\quad &\textnormal{col}(\mathbf{u}_{\textnormal{ini}},\hat{\mathbf{u}},\hat{\mathbf{y}})=\textnormal{col}(U,Y_f)\,g,\ \hat{\mathbf{u}} \in \mathcal{U},\ \hat{\mathbf{y}} \in \mathcal{Y},
    \end{aligned}
    \label{eqn:deepc}%
\end{equation}
where $Y_p$ is the first $L_0\cdot n_y$ rows of $Y$, $Y_f$ is the last $L'\cdot n_y$ rows of $Y$, and $p=1\text{ or }2$ depending on the noise model. This is known as the regularized DeePC algorithm (\cite{Coulson_2019_reg}). However, this method requires tuning of the hyperparameters $\lambda_g,\lambda_y$, and the input-output mapping obtained from the regularized problem is affected by the control cost $J_{\text{ctr}}$.

In this paper, we investigate an alternative approach to obtain a tuning-free, statistically optimal, and well-defined input-output mapping with noise-corrupted data by maximum likelihood estimation. This input-output mapping, known as the signal matrix model (\cite{yin2020maximum}), can then be applied in place of constraint (\ref{eqn:4b}) in the optimization problem when noisy data are available.

Specifically, consider the case where both the signal matrix and online measurements of $\mathbf{y}_{\text{ini}}$ are corrupted with i.i.d. Gaussian output noise:
\begin{equation}
\begin{cases}
    y_i^d = y_i^{d,0}+w_i^d,&(w_i^d)_{i=0}^{N-1}\sim \pazocal{N}(0,\sigma^2\mathbb{I}),\\
    \mathbf{y}_{\text{ini}} = \mathbf{y}_{\text{ini}}^{0}+\mathbf{w}_{p},& \mathbf{w}_{p}\sim \pazocal{N}(0,\sigma_{p}^2\mathbb{I}),
\end{cases}
\end{equation}
where for simplicity of exposition, we assume $n_y=1$. Define
$
    \tilde{\mathbf{y}}=\textnormal{col}\left(\epsilon_y,\hat{\mathbf{y}}\right)=Yg-\textnormal{col}\left(\mathbf{y}_{\text{ini}},\mathbf{0}\right),
$
where $\epsilon_y:=Y_pg-\mathbf{y}_{\text{ini}}$ is the total deviation of the past output trajectory from the noise-free case. Then the parameter $g$ can be estimated by maximizing the conditional probability density of observing the realization $\tilde{\mathbf{y}}$ corresponding to the available data given $g$, i.e., $\text{max}_g\ p(\tilde{\mathbf{y}}|g)$. This is equivalent to the following optimization problem:
\begin{equation}
    \underset{g\in\pazocal{G}}{\text{min}}\ \text{logdet}(\Sigma_y(g))+\begin{bmatrix}Y_pg-\mathbf{y}_{\text{ini}}\\\mathbf{0}\end{bmatrix}^\mathsf{T}\Sigma_y^{-1}(g)\begin{bmatrix}Y_pg-\mathbf{y}_{\text{ini}}\\\mathbf{0}\end{bmatrix}.
\label{eqn:opt0}
\end{equation}
where $\pazocal{G}$ is the parameter space defined by the noise-free input trajectory, namely
$
    \pazocal{G} = \{g\in \mathbb{R}^M\left|\,Ug=\textnormal{col}\left(\mathbf{u}_{\text{ini}},\hat{\mathbf{u}}\right)\right.\},
$
and $\Sigma_y(g)$ is the covariance of $\tilde{\mathbf{y}}$ given $g$ with entries
\begin{equation}
    \left(\Sigma_y(g)\right)_{i,j}=\sigma^2\sum_{k=1}^{M-|i-j|}g_k g_{k+|i-j|}+\begin{cases}\sigma_p^2,&i=j\leq L_0\\0,&\text{otherwise}\end{cases}.
    \label{eqn:py}
\end{equation}
To find a computationally efficient algorithm to solve (\ref{eqn:opt0}), we relax the problem by neglecting the cross-correlation between elements in the signal matrix, and thus the off-diagonal elements in $\Sigma_y(g)$. This leads to the following iterative algorithm:
\begin{equation}
    g^{k+1}=\text{arg}\underset{g\in\pazocal{G}}{\text{min}}\ \lambda(g^k)\norm{g}_2^2+\norm{Y_pg-\mathbf{y}_{\text{ini}}}_2^2,\quad \lambda(g^k) = \dfrac{L'\sigma_p^2}{\norm{g^k}_2^2}+L \sigma^2.
\label{eqn:optsqp}
\end{equation}
For details on the derivation of the maximum likelihood estimator and the associated iterative algorithm, readers are referred to Section IV in \cite{yin2020maximum}.

The maximum likelihood signal matrix model can be compactly represented as follows:
\begin{equation}
    \hat{\mathbf{y}}_\text{SMM}=Y_f\, g_\text{SMM}(\hat{\mathbf{u}};\mathbf{u}_{\text{ini}},\mathbf{y}_{\text{ini}},U,Y_p),
    \label{eqn:smm}
\end{equation}
where $g_\text{SMM}(\cdot)$ is the converged solution of the iterative algorithm (\ref{eqn:optsqp}). This model provides a data-driven input-output mapping when signals are measured with i.i.d. Gaussian output noise. Similar signal matrix models can be developed for alternative noise models, such as process noise.

\section{Data-Driven Predictive Control with the Signal Matrix Model}

In this section, the signal matrix model is applied to the data-driven predictive control framework as the predictor. Throughout this section, numerical examples tested on the following fourth-order LTI system:
\begin{equation}
    G(z) = \dfrac{0.1159(z^3+0.5z)}{z^4-2.2z^3+2.42z^2-1.87z+0.7225}
    \label{eqn:sys1}
\end{equation}
are considered. The goal is to achieve optimal output tracking of a sinusoidal reference trajectory $r_t = 0.5\sin\left(\frac{\pi}{10} t\right)$. The control design parameters are chosen as $Q=1$, $R=1$, $L'=10$, and $L_0=n_x=4$ with no input and output constraints, i.e., $\mathcal{U}=\mathcal{Y}=\mathbb{R}^{L'}$, unless otherwise specified. The length of the control task $N_c$ is 120. The offline input-output trajectory data are collected with unit i.i.d. Gaussian input signals, and the noise levels $\sigma^2$ and $\sigma_p^2$ are assumed known. To benchmark the performance, we consider the ideal deterministic MPC algorithm with knowledge of the true state-space model and the noise-free state measurements, which will be indicated by \textit{MPC} in the following figures.

The signal matrix model (\ref{eqn:smm}) can be directly incorporated in the receding horizon predictive control scheme (\ref{eqn:pc}) by replacing (\ref{eqn:4b}) with the SMM. However, the maximum likelihood estimator $g_\text{SMM}(\cdot)$ involves an iterative algorithm which does not have a closed-form solution. This makes the optimization problem computationally challenging to solve, especially in an online scheme. To address this problem, we propose warm-starting the iterative algorithm with the value of $g$ at the previous time instant and approximating $g_\text{SMM}(\cdot)$ with the first iterate, i.e.,
\begin{equation}
    g^t(\mathbf{\hat{u};\mathbf{u}_{\text{ini}},\mathbf{y}_{\text{ini}}},U,Y_p,g^{t-1})=\pazocal{P}(g^{t-1})\,\mathbf{y}_{\text{ini}}+\pazocal{Q}(g^{t-1})\,\textnormal{col}\left(\mathbf{u}_{\text{ini}},\hat{\mathbf{u}}\right),
\end{equation}
where, with an abuse of notation, $g^t(\cdot)$ denotes the approximation of $g_{\text{SMM}}(\cdot)$ with one iteration at time instant $t$, and the right hand side is the closed-form solution of the quadratic program (\ref{eqn:optsqp}) with $g^k=g^{t-1}$. This leads to the linearized signal matrix model:
\begin{equation}
    \hat{\mathbf{y}}_\text{LSMM}=Y_f\left(\pazocal{P}(g^{t-1})\,\mathbf{y}_{\text{ini}}+\pazocal{Q}(g^{t-1})\,\textnormal{col}\left(\mathbf{u}_{\text{ini}},\hat{\mathbf{u}}\right)\right),
    \label{eqn:lsmm}
\end{equation}
and thus the predictive control problem (\ref{eqn:pc}) becomes a quadratic program:
\begin{equation}
  \begin{aligned}\underset{\hat{\mathbf{u}},\hat{\mathbf{y}}}{\textnormal{minimize}}\quad & J_{\textnormal{ctr}}\\
    \textnormal{subject to}\quad &\hat{\mathbf{y}}=Y_f\left(\pazocal{P}(g^{t-1})\,\mathbf{y}_{\text{ini}}+\pazocal{Q}(g^{t-1})\,\textnormal{col}\left(\mathbf{u}_{\text{ini}},\hat{\mathbf{u}}\right)\right),\hat{\mathbf{u}} \in \mathcal{U}, \hat{\mathbf{y}} \in \mathcal{Y}.
    \end{aligned}
    \label{eqn:smmpc}%
\end{equation}
The value of $g$ is initialized with the pseudo-inverse solution to the reference output trajectory, i.e., $g^{\text{ini}}=\left[\text{col}\left(U_p,Y_p,Y_f\right)\right]^\dagger\text{col}\left(\mathbf{u}_0,\mathbf{y}_0,\mathbf{r}_0\right)$, where $\mathbf{u}_0=(u_i)_{i=-L_0}^{-1}$, $\mathbf{y}_0=(y_i)_{i=-L_0}^{-1}$, and $\mathbf{r}_0=(r_i)_{i=0}^{L'-1}$.

To assess the validity of the linearized SMM, we investigate the discrepancy between the linearized SMM and the iterative nonlinear SMM in the following example.
\begin{ex}
    Consider input-output trajectory data of length $N=50$ with noise level $\sigma^2=0.1$. The online measurement noise level is $\sigma_p^2=0.1$. The data-driven predictive control is conducted by solving (\ref{eqn:smmpc}) with the linearized SMM constraint. At each time instant $t=0,1,\dots,N_c-1$, the results of the linearized SMM predictor (\ref{eqn:lsmm}) and the iterative SMM predictor (\ref{eqn:smm}) are compared with the newly optimized input sequence and quantify their discrepancy by
    \begin{equation}
        \text{E} = \frac{\norm{\hat{\mathbf{y}}_\text{LSMM}-\hat{\mathbf{y}}_\text{SMM}}_2}{\norm{\hat{\mathbf{y}}_\text{SMM}}_2}.
    \end{equation}
    The histogram of $E$ is plotted in Figure~\ref{fig:1}.
    \begin{figure}[htb]
        \centering
        \includegraphics{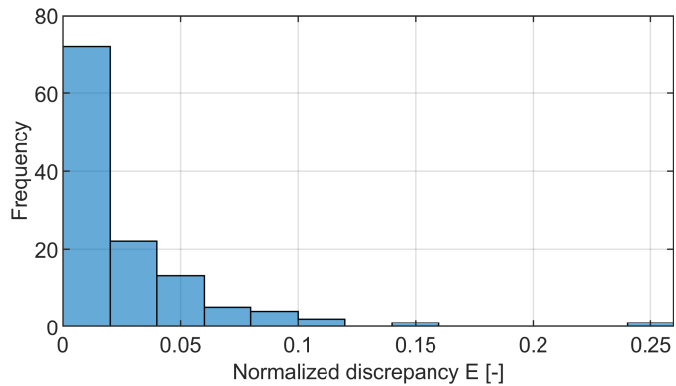}
        \vspace{-1em}
        \caption{Normalized discrepancy between the linearized SMM and the iterative SMM.}
        \label{fig:1}
    \end{figure}
    \label{ex:1}
\end{ex}
\vspace{-1em}
It can be seen from Figure~\ref{fig:1} that the error resulting from the linearization of SMM is less than 5\% for most time instants, which shows that the linearized SMM can obtain an accurate input-output mapping. Therefore, in the following analysis, the linearized SMM (\ref{eqn:lsmm}) will be employed.

In the first test, we compare the signal matrix model predictive control (SMM-PC) with the regularized DeePC algorithm in the following example.
\begin{ex}
    Consider an MPC problem with the same trajectory data and online measurement noise as in Example~\ref{ex:1}. In addition to SMM-PC, the regularized DeePC algorithm (\ref{eqn:deepc}) is also applied. The hyperparameter $\lambda_g$ is selected from a nine-point logarithmic grid between 10 and 1000 with an oracle (choosing the value with the optimal performance a posteriori); $\lambda_y$ is fixed to 1000 since the control performance is not sensitive to it in this example. 50 Monte Carlo simulations are conducted, and the closed-loop input-output trajectories within one standard deviation are plotted in Figure~\ref{fig:2}.
    \begin{figure}[htb]
        \centering
            \includegraphics{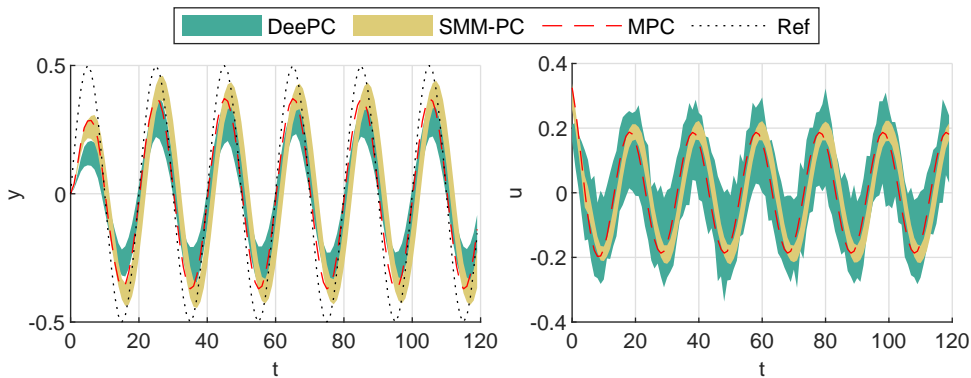}
        \vspace{-1em}
        \caption{Closed-loop input-output trajectories with regularized DeePC and linearized SMM.}
        \label{fig:2}
    \end{figure}
\end{ex}
\vspace{-1em}
As shown in Figure~\ref{fig:2}, the output trajectory obtained with the SMM-PC algorithm is closer to the reference trajectory with smaller inputs compared to the regularized DeePC algorithm. In addition, the applied input trajectory is much smoother for the SMM-PC algorithm, which indicates that the input-output mapping predicted by the linearized SMM is more accurate compared to the regularization method. For more detailed comparison with the DeePC algorithm, readers are referred to Section~VI.C in \cite{yin2020maximum}.

The data-driven formulation of the input-output mapping with the linearized signal matrix model enables flexible extensions to the algorithm under multiple scenarios, including large datasets, poor data quality, and time-varying systems. These extensions will be discussed as follows.

When a large set of input-output trajectory data are collected, i.e., $N\gg L$, the dimension of the intermediate parameter $g\in \mathbb{R}^M$ will be much larger than the time horizon of the predictive control. This is an undesired feature as it leads to high computational complexity online. Therefore, we propose the following strategy to compress the dimension of $g$ to $2L$ regardless of the data length.
\begin{prop}
    Let the singular value decomposition (SVD) of the signal matrix be $\textnormal{col}\left(U,Y\right)=WSV^\mathsf{T}\in\mathbb{R}^{2L\times M}$. Define the compressed signal matrix as $\textnormal{col}\,(\tilde{U},\tilde{Y})=WS_{2L}\in\mathbb{R}^{2L\times 2L}$, where $S_{2L}$ is the first $2L$ columns of $S$. Then the SMM-PC algorithm with the signal matrix $\textnormal{col}\left(U,Y\right)$ is equivalent to that with the compressed signal matrix $\textnormal{col}\,(\tilde{U},\tilde{Y})$.
    \label{prop:2}
\end{prop}
\begin{proof}
    This is a direct extension of Proposition 1 in \cite{yin2020maximum}.
\end{proof}
Proposition~\ref{prop:2} shows that regardless of the size of the dataset, an SVD operation can be conducted offline to reduce the size of the signal matrix to a square matrix while obtaining exactly the same results as with the original signal matrix. In this way, the online computational complexity is only dependent on the control horizon.

\subsection{Incorporation of Online Data}

In addition to the offline data $(u_i^d, y_i^d)_{i=0}^{N-1}$, online measurements can also be used to construct the signal matrix. By doing so, the knowledge of the unknown system improves as the predictive controller is deployed. This is particularly useful when the quality of the offline data is poor or the model parameters are varying online. This strategy can be interpreted as an equivalent of adaptive MPC (\cite{Fukushima_2007}) for the signal matrix model.

At each time instant $t\geq L$, a new column can be added to the signal matrix with
\begin{equation}
    \begin{bmatrix}U_{t+1}\\Y_{t+1}\end{bmatrix}=\begin{bmatrix}
    \gamma U_{t}&(u_i)_{i=t-L+1}^{t}\\\gamma Y_{t}&(y_i)_{i=t-L+1}^{t}
    \end{bmatrix},
    \label{eqn:online}
\end{equation}
where $\text{col}\left(U_t,Y_t\right)$ is the adaptive signal matrix applied to the SMM-PC algorithm at time $t$, and $\gamma$ is the forgetting factor of previous trajectories. The factor $\gamma$ is chosen as 1 for LTI systems and $\gamma\in(0,1)$ when model parameter variations are expected. Similar to the offline signal matrix, the adaptive signal matrix can also be compressed online to maintain $2L$ columns. The incremental singular value decomposition algorithm in \cite{Brand_2002} can be applied to reduce the computational complexity of updating the compressed signal matrix, which makes use of the property that the SVD of $\text{col}\left(U_t,Y_t\right)$ in (\ref{eqn:online}) is already known from the previous timestep.

In the following two examples, the effects of incorporating online data are investigated under high noise levels and with slowly varying parameters respectively.
\begin{ex}
    Consider the case where the output measurements are very noisy with $\sigma^2=\sigma_p^2=1$. The offline data length is $N=100$. The SMM-PC algorithms are compared with the fixed signal matrix $\textnormal{col}\left(U,Y\right)$ and the adaptive signal matrix $\textnormal{col}\left(U_t,Y_t\right)$ with $\gamma=1$. The deviations of the closed-loop trajectories from the ideal MPC trajectory are plotted on the left of Figure~\ref{fig:3}. In addition, 200 Monte Carlo simulations are conducted, and the boxplot of the total cost $J_\textnormal{tot}$ is shown on the right of Figure~\ref{fig:3}.
    \begin{figure}[htb]
        \centering
            \includegraphics{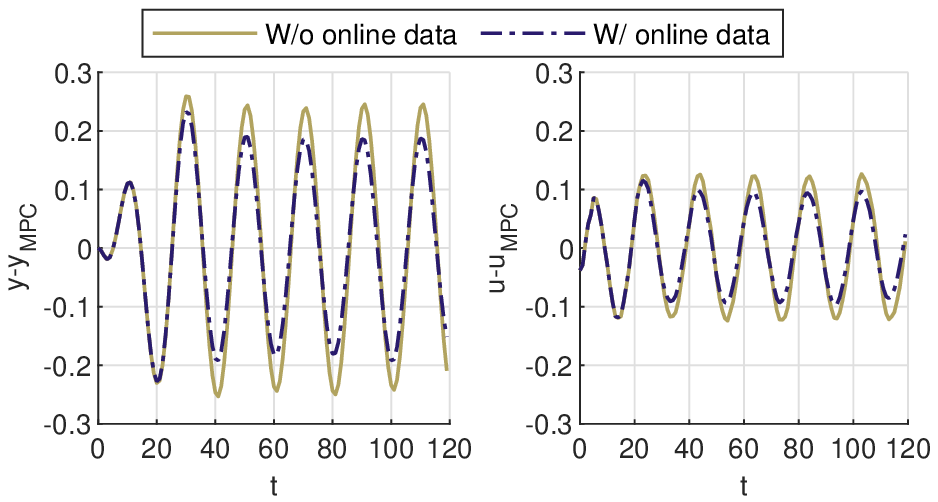}
            \includegraphics{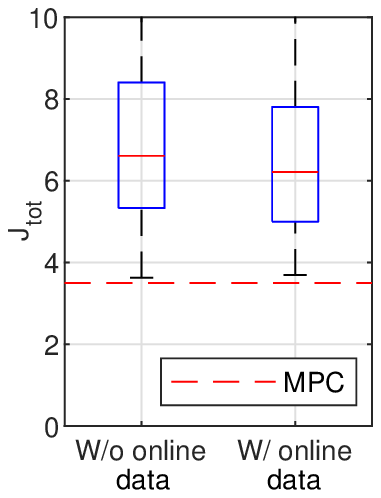}
        \vspace{-1em}
        \caption{Effects of online data adaptation in SMM-PC for datasets with high noise levels. Left: deviation from ideal MPC, right: boxplot of total cost $J_\textnormal{tot}$.}
        \label{fig:3}
    \end{figure}
    \label{ex:3}
\end{ex}
\vspace{-1em}
\begin{ex}
    Consider the case where one of the model parameters drifts slowly online with 
    \begin{equation}
        G(z) = \dfrac{0.1159(z^3+0.5z)}{z^4-2.2z^3+2.42z^2-\theta(t)z+0.7225},\, \theta(t) = \dfrac{1.87}{1+t/1500}.
    \end{equation}
    The offline data length and the noise levels are $N=50$, $\sigma^2=\sigma_p^2=0.01$ respectively. The SMM-PC algorithms are compared with the fixed signal matrix $\textnormal{col}\left(U,Y\right)$ and the adaptive signal matrix $\textnormal{col}\left(U_t,Y_t\right)$ with $\gamma=1,0.9,0.7,0.5$. The stage costs $J_t$ are plotted on the left of Figure~\ref{fig:4}, and the boxplot of the total cost $J_\textnormal{tot}$ is shown on the right of Figure~\ref{fig:4} with 50 Monte Carlo simulations.
    \begin{figure}[htb]
        \centering
            \includegraphics{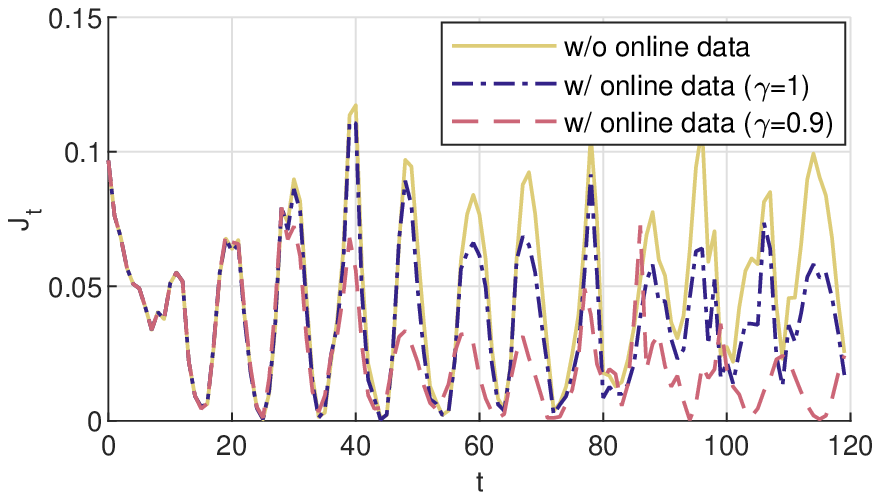}
            \includegraphics{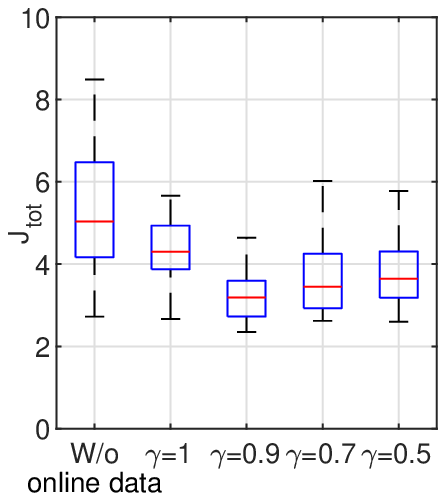}
        \vspace{-1em}
        \caption{Effects of online data adaptation in SMM-PC for slowly time-varying systems. Left: stage cost $J_t$, right: boxplot of total cost $J_\textnormal{tot}$.}
        \label{fig:4}
    \end{figure}
    \label{ex:4}
\end{ex}
\vspace{-1em}
As can be seen from Examples~\ref{ex:3}, the control performance improves by using adaptive signal matrices, as more online data are accumulated, though it cannot converge to the ideal MPC performance since the noise in $\mathbf{y}_{\text{ini}}$ still exists. For Example~\ref{ex:4}, the controller further benefits from a forgetting factor $\gamma<1$ to reduce the weight of previous trajectories with less accurate model parameters. The best performance is achieved with $\gamma=0.9$. These results demonstrate that incorporating online data is effective in achieving better control performance for both high noise levels and slowly time-varying systems.

\subsection{Initial Condition Estimation}

Another problem of measurements with high noise levels is the accuracy of initial condition estimation with the past trajectory measurements $\text{col}\left(\mathbf{u}_{\text{ini}}, \mathbf{y}_{\text{ini}}\right)$. Although the large prediction errors due to poor offline measurements can be compensated by online data adaptation as shown in the previous subsection, the validity of the linearized SMM still suffers from deviated initial conditions. In the noise-free case, Proposition~\ref{prop:1} holds for any $L_0\geq n_x$, which gives the exact initial condition of the trajectory. So it is always reasonable to select $L_0 = n_x$ to reduce computational complexity. This is also the choice in previous examples. However, in the presence of noise, the signal matrix model only gives an estimate of the initial condition, whose accuracy depends on the choice of $L_0$. Therefore, when abundant data are available but the noise level is high, it can be beneficial to select $L_0>n_x$ to obtain a more accurate initial condition estimation, as shown in the following example.
\begin{ex}
    Consider the same trajectory data and online measurement noise as in Example~\ref{ex:3}. The controller performance of SMM-PC is compared with $L_0=n_x=4$ and $L_0=10$. As another comparison, we introduce an additional MPC algorithm with a finite impulse response model obtained by simulating the signal matrix model with an impulse of $\mathbf{u}_{\text{ini}}=\mathbf{0},\,\mathbf{y}_{\text{ini}}=\mathbf{0},\,\mathbf{u}=\textnormal{col}(1,\mathbf{0}),\,\sigma_p=0$. This impulse MPC algorithm does not require the noisy past output measurements, so it circumvents the initial condition estimation problem. The stage costs $J_t$ are plotted in Figure~\ref{fig:5}. 
    \begin{figure}[htb]
        \centering
            \includegraphics{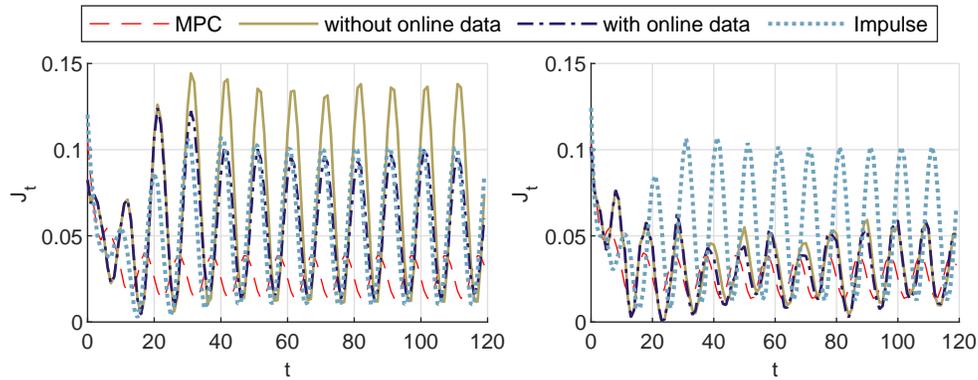}
        \vspace{-1em}
        \caption{Stage costs with different $L_0$. Left: $L_0=4$, right: $L_0=10$.}
        \label{fig:5}
    \end{figure}
\end{ex}
\vspace{-1em}
Although the impulse MPC algorithm only uses a fixed predictor identified offline, the results from Figure~\ref{fig:5} demonstrate that when $L_0=4$, SMM-PC performs worse than impulse MPC without online data and similarly when online data are incorporated, due to inaccurate initial condition estimation. When a larger value of $L_0=10$ is selected, the SMM-PC performs significantly better than impulse MPC and close to the ideal MPC algorithm. 

\subsection{Regularized Signal Matrix Model Predictive Control}
In addition to the input-output mapping, the signal matrix model also gives an uncertainty model of the output prediction with $\text{cov}(\hat{\mathbf{y}}|g)=\Sigma_{y_f}$, where $\Sigma_{y_f}$ denotes the last $L'$ rows and columns of $\Sigma_y$. So the variance of the output prediction is given by $\left(\Sigma_{y_f}\right)_{i,i}=\sigma^2\norm{g}_2^2$. In order to enhance the robustness of the SMM-PC algorithm, an additional regularization term can be introduced to the objective function:
\begin{equation}
    J_{\text{Reg-SMM}} = J_{\text{ctr}} + \zeta\cdot \sigma^2\norm{g^t}_2^2.
\end{equation}
Despite their similar forms, this regularization term is different from the $\lambda_g$-term in regularized DeePC, where the regularization is required to fix the underdetermined input-output mapping. In regularized SMM-PC, the input-output mapping is well-defined via the linearized SMM constraint (\ref{eqn:lsmm}) and the regularization term is used to penalize trajectories that result in large prediction errors. Tuning of the hyperparameter $\zeta$ is also easier than regularized DeePC since the dependence of the noise level $\sigma^2$ is already separated in the prediction error formulation.

In the following example, the performance of SMM-PC is compared for different choices of $\zeta$.
\begin{ex}
    Consider the same data generation parameters as in Example~\ref{ex:1} ($N=50,\sigma^2=\sigma_p^2=0.1$). The regularized SMM-PC algorithm is applied for $\zeta=0,1,10,10^2,10^3,10^4$ with 50 Monte Carlo simulations each. The boxplot of average $\norm{g^t}_2^2$ is shown in Figure~\ref{fig:6}(a) as a measure of prediction errors. The boxplots of the total cost $J_\textnormal{tot}$ and the input cost $J_\textnormal{tot}^u:=\sum_{t=0}^{N_c-1}\norm{u_t}_R^2$ are plotted in Figures~\ref{fig:6}(b) and (c) respectively.
    \begin{figure}[htb]
        \centering
        \includegraphics{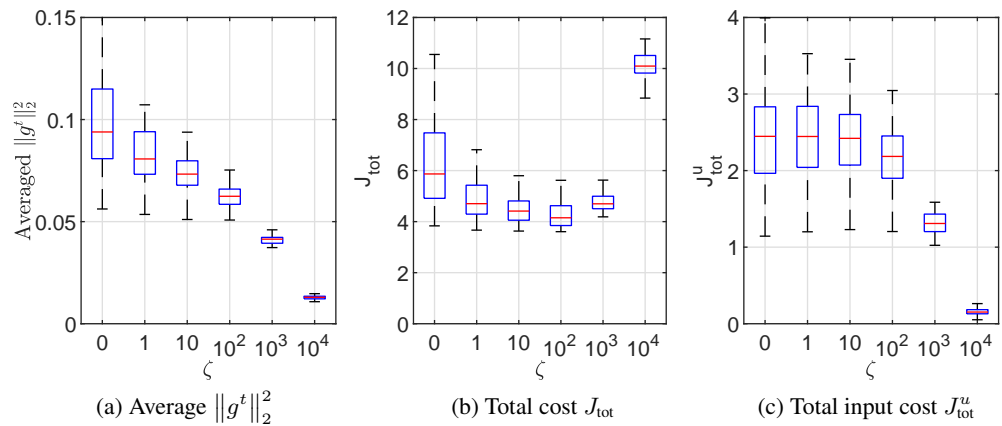}\\
        \footnotesize \hspace{2em} (a) Average $\norm{g^t}_2^2$ \hspace{7em} (b) Total cost $J_\textnormal{tot}$ \hspace{5.5em} (c) Total input cost $J^u_\textnormal{tot}$ 
        \caption{Performance of regularized SMM-PC with $\norm{g}_2^2$-regularization.}
        \label{fig:6}
    \end{figure}
\end{ex}
\vspace{-1em}
As can been seen from Figure~\ref{fig:6}, by increasing the hyperparameter $\zeta$, the SMM-PC algorithm tends to select trajectories with smaller prediction errors, which also correspond to more conservative input trajectories with smaller magnitudes. As a result, the average performance first increases due to better prediction and then decrease due to large output tracking errors with the increase of $\zeta$.

\section{Conclusions}

In this paper, a novel data-driven predictive control algorithm with the signal matrix model is investigated. The algorithm uses a linearized SMM to provide a close approximation to the maximum likelihood estimator of input-output trajectories. It combines the benefits of model-based and model-free methods by providing a non-parametric input-output mapping to characterize the system. By using the linearized SMM to implicitly parametrize future trajectories, the algorithm and its regularized version show good performance in complex scenarios with large datasets, poor data quality, parameter variation.

\acks{This work was supported by the Swiss National Science Foundation under Grant 200021\_178890.}

\bibliography{refs}

\end{document}